\documentclass[letterpaper, 10pt, conference]{ieeeconf}      

\IEEEoverridecommandlockouts                              

\usepackage{graphics} 
\usepackage{epsfig} 
\usepackage{mathptmx} 
\usepackage{times} 
\usepackage{amsmath} 
\usepackage{amssymb}  
\usepackage{graphicx}
\usepackage{xcolor}
\usepackage{booktabs}

\title{\LARGE \bf
Nonlinear control for an uncertain electromagnetic actuator 
}

\author{Flavien Deschaux\thanks{F. Deschaux is with CNRS, LAAS, Univ de Toulouse, INSA, France, e-mail: fdeschau@laas.fr.}, Frederic Gouaisbaut\thanks{F. Gouaisbaut is with CNRS, LAAS, Univ de Toulouse, UPS, France, e-mail: fgouaisb@laas.fr.} and Yassine Ariba \thanks{Y. Ariba is with the school of engineering Icam, Toulouse, France, e-mail: yassine.ariba@icam.fr and with CNRS, LAAS,  Univ de Toulouse, Toulouse, France.}}

\begin{document}

\newtheorem{theorem}{Theorem}
\newtheorem{Proof}{Proof}
\newtheorem{Remark}{Remark}
\renewcommand{\bar}{\overline}

\maketitle
\thispagestyle{empty}
\pagestyle{empty}

\begin{abstract}

This paper presents the design of a nonlinear control law  for a typical electromagnetic actuator system. Electromagnetic actuators are widely implemented in industrial applications, and especially as linear positioning system. In this work, we aim at taking into account a magnetic phenomenon that is usually neglected: flux fringing. This issue is addressed with an uncertain modeling approach. The proposed control law consists of two steps, a backstepping control regulates the mechanical part and a sliding mode approach controls the coil current and the magnetic force implicitly. An illustrative example shows the effectiveness of the presented approach.
\end{abstract}

\section{Introduction}
For many years Electro-Magnetic Actuators (EMA) have been developed and used in industrial environment, and especially in automotive industries \cite{ito1980electronic}. This interest can be easily explained by several factors: their small size, their simple structure, their cost/efficiency ratio and the very large range of applications.  For example, such technology is used at a micro-scale with the micro-electro-mechanical-systems (MEMS) \cite{cugat2003magnetic}. At a macro-scale, magnetic bearing systems \cite{di2007model}, electromagnetic positioning systems \cite{su2007towards} or electronic injection systems of thermal motors \cite{ito1980electronic} are all typical examples of EMA. Also, at a larger scale, this technology is under development for magnetic levitated vehicle \cite{jeong2017analysis}. For several years, the National Centre for Space Studies (CNES\footnote{The CNES, meaning Centre National d'Etudes Spatiales, is the French space agency.}) investigates innovative technologies to expand the use of electrical actuators in Ariane launchers. More specifically, the space agency has been working with CSTM, a mechanical engineering company, to replace pneumatic valves with electromagnetic actuators \cite{publi_rome_cnes}. The present study continues this work, in collaboration with the CNES and CSTM, and focuses on the control issue.

EMA are often controlled by linear control strategies such as  Proportional Derivative controller $(PD)$  \cite{siddiqui2017stabilizing},  Proportional Integral Derivative controller $(PID)$ \cite{endalecio2017position} or Linear Parameter Varying $(LPV)$ \cite{yaseen2017comparative}, Model Predictive Control $(MPC)$ \cite{di2007model} and Linear Quadratic Regulator $(LQR)$ controllers  \cite{forrai2007electromagnetic}. In \cite{siddiqui2017stabilizing} and \cite{di2007model}, a common approximation of the electromagnetic force $ F_{mag} = \frac{B^2S}{2\mu_0}$ is considered. Such simplification may be valid when the magnetic circuit is neglected and only airgap surface are considered. A linearization around a settling point leads to the use of a second order transfer function between the input current and the position of the EMA. In \cite{yaseen2017comparative} and \cite{forrai2007electromagnetic}, the expression of the magnetic force has been refined, as $ F_{mag} = N^2\frac{i^2 k_2}{(k_0 + k_1 x)^2}  $ even though the model is, eventually linearized. In this expression, $ F_{mag}$ depends explicitly on the airgap  $x$ and the actuator current $i$.  $N$ and $k_i$ are constants related to the physical structure of the system. The nature of this latter expression has led researches to design non-linear control laws.

Hence, in \cite{schwarzgruber2012nonlinear}, a backstepping approach is proposed to address the stabilization of a nonlinear model of the EMA. Furthermore, a worst-case estimation of the desired magnetic force provides conditions for the backstepping gains to avoid saturation of the magnetic force. However, notice that no explicit analytical model for the electromagnetic force has been developed. Instead, two look-up tables giving the relationships between the magnetic flux, the position and the magnetomotive force on one hand, and between the position, the magnetomotive force and the magnetic force, on the other hand, are used. The approach thus requires finite element method simulation before the control design. In \cite{mercorelli2012antisaturating}, the author uses an adaptive preaction (feedforward) to charge the coil energetically in order to compensate the spring force and a sliding mode control strategy to avoid saturation and achieve a soft landing control. A magnetic equivalent circuit of the actuator is designed in order to evaluate the magnetic force but some terms are neglected like the product between the inductance and the current time derivative $L(t)\frac{di(t)}{dt}$ used in the electrical modeling. \cite{liu2003nonlinear} combines a linear and nonlinear controls, the magnetic force is regarded as a virtual control input: a linear dynamic output feedback is used to construct the image of the magnetic force needed to stabilize the system and a backstepping method is applied to find the input voltage required. However, the proposed model is simplified and does not contain non-controlled forces like spring force or gravity. \cite{mercorelli2012antisaturating} and \cite{schwarzgruber2012nonlinear} take into account the magnetic saturation phenomenon, but in the literature the flux fringing effect has not been investigated so far.

The flux fringing is a phenomenon that occurs when the magnetic flux flows from a ferromagnetic material to the air. The magnetic flux tends to expand before being canalized again when it flows back through the material. The larger the airgap is, the larger the equivalent surface is.  The EMA principle being based on airgap control, if the stroke is significant with respect to the magnetic circuit surface, flux fringing must be taken into account. In control literature, the surface in the airgap is often considered equal to the magnetic circuit surface. 

A first objective of this paper is to consider a more comprehensive model of the actuator and to take into account the flux fringing. The analytical model is derived from a reluctance network approach. This model is then supported by finite element method simulations with COMSOL \cite{Comsol} and electrical system based simulations with PLECS \cite{PLECS}. The effects of flux fringing are then embedded into an uncertain model. Taking into account this issue leads to an uncertain magnetic force, and therefore a non-linear uncertain model. Then a backstepping control cannot compensate a nonvanishing term depending on a spring force. The same applies for the sliding mode, the non-vanishing term does not satisfy the matching condition. The proposed control law is then developed in two steps in order to stabilize the system despite the presence of uncertain parameters. Firstly, the backstepping method is used to stabilize the position and the speed of the moving part of the actuator by computing a suitable coil current signal. Secondly, a sliding mode control is used to control this latter variable. The sliding surface represents the difference between the actual current in the actuator and the desired current for the backstepping control. An illustrative example shows that the proposed approach is able to ultimately bound the EMA despite the model uncertainties.


\section{System description and modeling}

\subsection{Description of the EMA}

A schematic representation of the 1-DOF positioning system is shown in Fig. \ref{fig-EPS}. This is a typical setup for electromagnetic valve actuators as presented in \cite{mercorelli2012antisaturating} and \cite{peterson2004extremum}. The electromagnetic system is composed of a multi-turn coil winding a magnetic circuit which is fixed to the frame. A silicon O-ring is installed for limited friction and for sealing. A spring is used to counteract the magnetic force and to ensure that the system returns to the closed position when no supplied.

\begin{figure}[h!]
\begin{center}
\includegraphics[height=4cm]{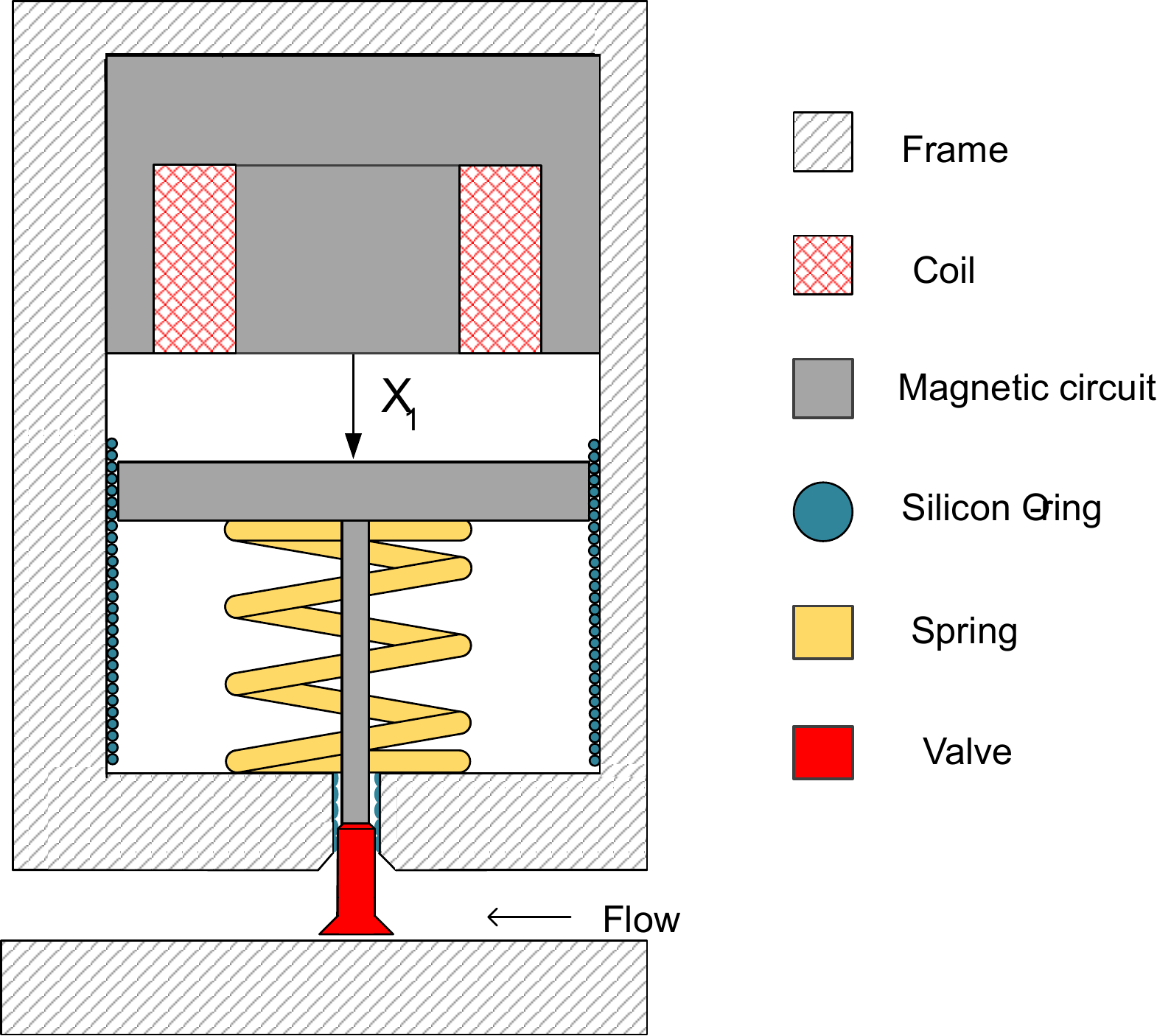} 
\end{center}
\caption{Schematic of the EMA}
\label{fig-EPS}
\end{figure}



\label{par2a}

In order to design a stabilizing control law for the EMA, we firstly develop a nonlinear model, which takes into account some new features like the flux fringing phenomenon. The following three sub-sections will develop the modeling steps to derive a non linear mathematical model. The EMA will be modeled  as a 3 dimensional system with the state vector $x = (x_1 ,\, x_2 ,\, x_3)^T$. The state components are the position $x_1$ of the mobile part, its velocity $x_2$ and the current of the coil $x_3$. The actuator input $u$ is the voltage at the terminals of the coil. 

\subsection{Electromagnetic part}

%
%
%
Following \cite{woodson1968electromechanical}, the electromagnetic energy is defined by 

\begin{equation}
W_{mag} = \frac{1}{2}L(x_1) x_3^2 \text{,}
\label{eq-energie_mag}
\end{equation}
where $L$ is the actuator inductance and depends on $x_1$. The magnetic force is therefore: 
\begin{equation}
F_{mag}(x_1,x_3) = \frac{1}{2}x_3^2 \frac{dL}{dx_1}\text{.} 
\label{eq-Fmag}
\end{equation}
Furthermore, the total inductance of the system is defined by \cite{woodson1968electromechanical}: 
\begin{equation}
L = \frac{N^2}{\rho(x_1)},
\label{eq-inductance}
\end{equation}
with $N$ the number of coil's turn and $\rho(x_1)$ the total reluctance of the magnetic circuit. The reluctance is defined by $\rho = \frac{l}{\mu S}$ \cite{woodson1968electromechanical}, where $l$ is  the length of the magnetic tube, $S$ its section and $\mu$ the permeability of the material.

Using a finite element method simulation of the magnetic field density in the actuator ($Comsol$ $Multiphysics$ \cite{Comsol}), the path of the magnetic field lines are depicted in Fig. \ref{fig-Comsol}. 
\begin{figure}[h!]
\begin{center}
\includegraphics[height=4cm]{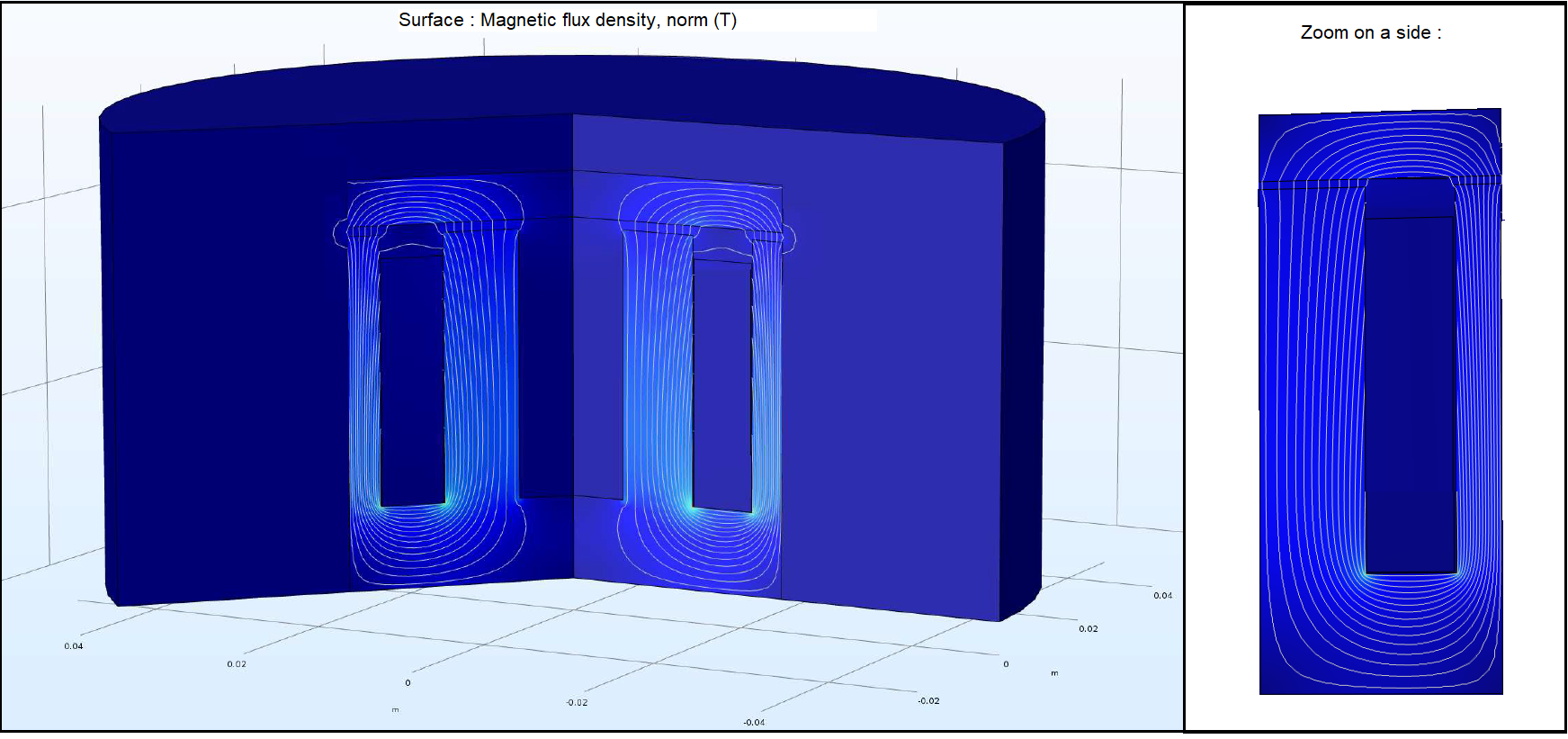} 
\end{center}
\caption{Comsol Simulation - Magnetic field lines}
\label{fig-Comsol}
\end{figure}

Notice that on Fig. \ref{fig-Comsol}, the flux fringing is visible in the airgap: some magnetic field lines do not take the shortest way between the mobile and the fixed part of the EMA. This last figure allows to elaborate an equivalent reluctance network following the magnetic field lines \cite{raminosoa2009reluctance} \cite{perho2002reluctance}. This network, given in Fig. \ref{fig-Schema_EPS}, is composed of six series reluctances delimited by straight sections: three for the body, two for the airgaps and one for the moving part.
All physical parameters used to compute the reluctances network are defined in the Fig. \ref{fig-Schema_EPS}. This choice of structure leads to a reluctance network depicted in Fig. \ref{fig-Schema_EPS}.
\begin{figure}[h!]
\begin{center}
\includegraphics[height=4.5cm]{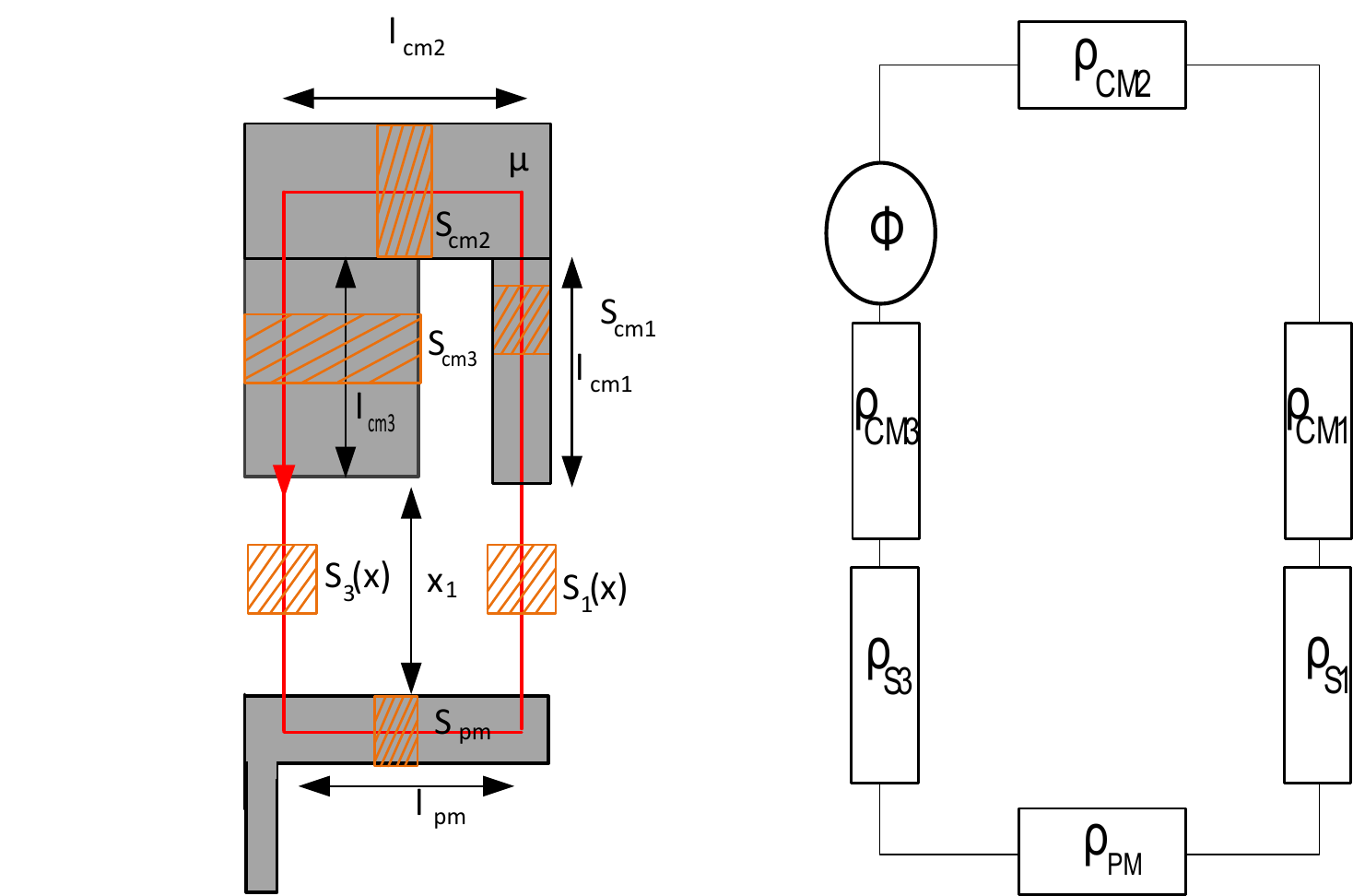} 
\caption{Representation of the physical parameters (left) and the associated reluctance network (right)}
\label{fig-Schema_EPS}
\end{center}
\end{figure}

Hence, the global reluctance of the system is computed by 
\begin{equation}
\begin{array}{r c l}
 \rho(x_1) & = & \displaystyle \frac{x_1}{\mu_0 S_1(x_1)} +\frac{x_1}{\mu_0 S_3(x_1)} + \rho_0 \text{,}\\		   
		   & = & \displaystyle \rho_x x_1 + \rho_0 \text{,}
\end{array}
\label{eq-Reluc-totale}
\end{equation}
with $\rho_0$ the sum of the magnetic circuit reluctances, and $\rho_x$ the sum of reluctance that depends on the airgap. The magnetic system of Fig. \ref{fig-Schema_EPS} has been simulated with the electrical engineering software $PLECS$. For nominal values, the analytical expression of $L$ in equations(\ref{eq-inductance})and (\ref{eq-Reluc-totale}) has been compared to numerical simulations. Fig. \ref{fig-Plecs} shows the computation of the actuator inductance.

\begin{figure}[h!]
\begin{center}
\includegraphics[height=3.2cm]{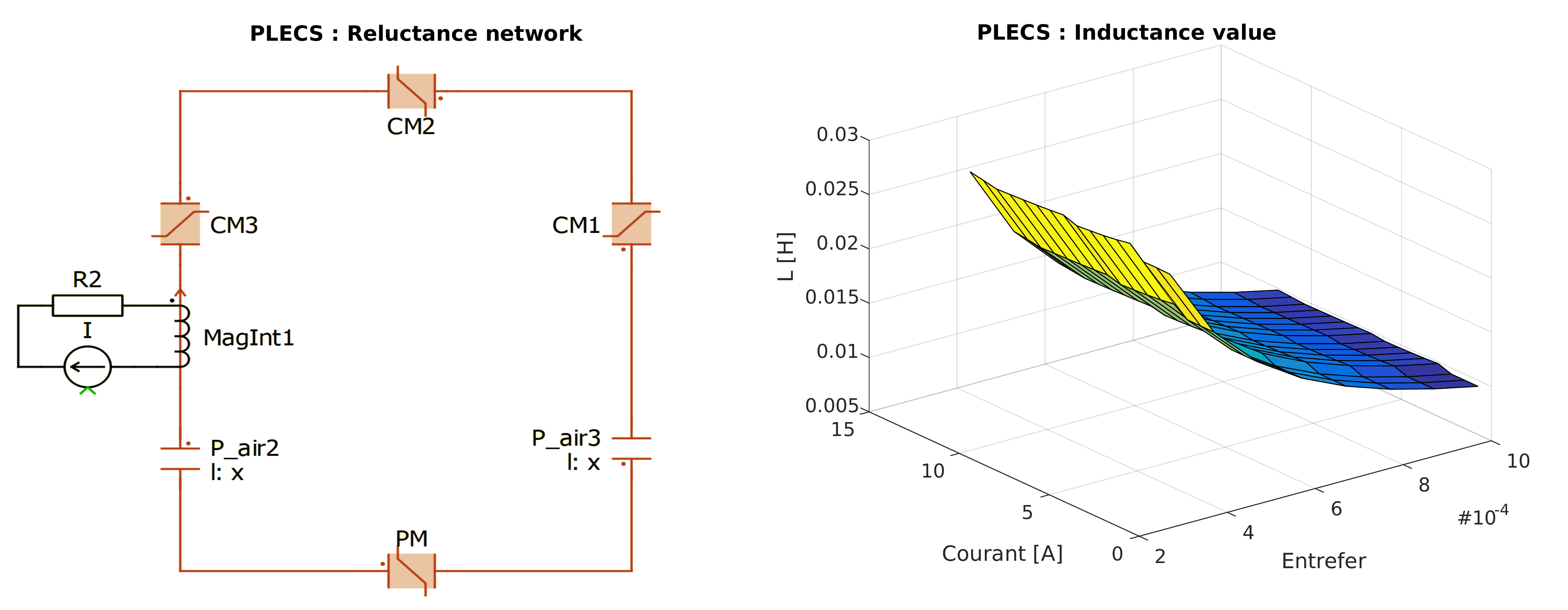} 
\caption{Reluctance network simulation using PLECS}
\label{fig-Plecs}
\end{center}
\end{figure}

The novelty of the model is to take into account a phenomenon arisen in the magnetic circuit called flux fringing:  the equivalent sections $S_i$ encompassing the flux in the airgap is not constant and depends on the airgap length $x_1$. Several models have been proposed in Bouchard \cite{bouchard1999electrotechnique}, Woodson \cite{woodson1968electromechanical} and Muhlethaler \cite{muhlethaler2012modeling}. Following Fig. \ref{fig-Comsol} and Fig. \ref{fig-Schema_EPS} two approaches are proposed. The first, most widely used, consists of an estimate of the expansion with a weighting coefficient such as $ S_{3} = \alpha_3 S_{CM3} $ and $ S_{1} = \alpha_1 S_{CM1} $. The second approach approximates the equivalent section by a function depending on the surface of the magnetic circuit and the airgap. For example, if the surface of the magnetic circuit is a rectangle, $S = ab$ , the surface in the airgap will be estimated as $ S_a = (a+x_1)(b+x_1) $. Note that those results are mostly empirical results and will serve as an adjustment variable for our model.

In order to cope with this physical phenomenon, airgap surfaces will be modeled as uncertain parameters. Let $\underline{S}_i$ and $\bar{S}_i$ be two positive constants bounding the surface $S_i(x_1) $, $ \displaystyle 0 < \underline{S}_i < S_i(x_1) < \bar{S}_i$. 
As a result $\rho(x_1)$ may be bounded by $\displaystyle  0 < \underline{\rho}(x_1) < \rho(x_1) < \bar{\rho}(x_1)$ with $\displaystyle \underline{\rho}(x_1) = \frac{x_1}{\mu_0 \bar{S}_1} +\frac{x_1}{\mu_0 \bar{S}_3} + \rho_0$ and $\displaystyle \bar{\rho}(x_1) = \frac{x_1}{\mu_0 \underline{S}_1} +\frac{x_1}{\mu_0 \underline{S}_3} + \rho_0 $.
In the same way, the total inductance defined in (\ref{eq-inductance}) $\displaystyle L(x_1) = \displaystyle \frac{N^2}{\displaystyle  \rho_x x_1 + \rho_0}$ is bounded by $\displaystyle  0 < \underline{L}(x_1) < L(x_1) < \bar{L}(x_1)$ with $\displaystyle \underline{L}(x_1) = \frac{N^2}{\bar{\rho}(x_1)}$ and $\displaystyle \bar{L}(x_1) = \frac{N^2}{\underline{\rho}(x_1)}$.
Finally, the expression of the magnetic force is obtained by:

\begin{equation}
 F_{mag}(x_1,x_3) = -\frac{1}{2} x_3^2 N^2 \frac{ \rho_x  }{\rho(x_1)^2} = -\frac{1}{2}x_3^2 \mu(x_1)\text{,}
 \label{eq-Fmag-end}
\end{equation}
with $\mu(x_1) = \displaystyle N^2 \frac{ \rho_x  }{\rho(x_1)^2} $. In the same way as the reluctance $\rho(x_1)$ and the inductance $L(x_1)$, define $\underline{\mu}(x_1) = N^2 \frac{ \rho_x  }{\bar{\rho}(x_1)^2} $ and $\bar{\mu}(x_1) = N^2 \frac{ \rho_x  }{\underline{\rho}(x_1)^2} $ such that $\mu(x_1)$ is bounded by  $ 0 < \underline{\mu}(x_1) < \mu(x_1) < \bar{\mu}(x_1)$.

\subsection{Electrical part}

Applying the input voltage $ u  $ at the terminals of the coil the electrical dynamic is described by:

\begin{equation}
u = Rx_3 + \frac{d\Phi}{dt}\text{,}
\end{equation} 
with $R$ the coil internal resistance. By definition of the magnetic flux $\Phi = Lx_3$ \cite{woodson1968electromechanical} we get: 

\begin{equation}
u = Rx_3 + L\frac{dx_3}{dt} + x_3\frac{dL}{dt} \text{.}
\end{equation} 
Since   $ \displaystyle \frac{dL}{dt} = \frac{\partial{L}}{\partial{x_1}}\frac{dx_1}{dt}$, a dynamical equation for the current $x_3$ is formulated as

\begin{equation}
\frac{dx_3}{dt} = \displaystyle \frac{1}{L(x_1)} \Big( u - Rx_3 + x_2x_3\mu(x_1)\Big)\text{.}
\label{eq-Elec}
\end{equation}

\subsection{Mechanical part}

The application of the Newton's second law to the moving part gives: 

\begin{equation}
m \frac{dx_2}{dt} = -F_{mag} + F_{ext}\text{,}
\label{eq-Newton}
\end{equation}
where $F_{ext} $ is the sum of the external forces: $F_{ext} =  F_{friction} + F_{spring}$.  $F_{friction}$ represents the friction force which is proportional to the speed,  $F_{friction}  =- \lambda x_2$ and $F_{spring}$ is the force due to the spring, proportional to the position, $F_{spring} = -K x_1$ where $\lambda$ and $K$ are positives constants.


\subsection{State space model}

In this work we consider a 3 dimensional model of the electromagnetic actuator. The state variables have been defined in Paragraph \ref{par2a}, and gathering equations (\ref{eq-Elec}) and (\ref{eq-Newton}), a state space model is obtained in (\ref{eq-systeme}). Note that the control input $u(t)$ has effect only on the third equation. We will name $x_0$ the initial state of the actuator.


\begin{equation}
   \left \{
   \begin{array}{r c l}
      \displaystyle \dot{x}_1  & = & \displaystyle x_2\text{,}   \\[4mm] 
      \displaystyle \dot{x}_2  & = & \displaystyle \frac{1}{m}\left[ \frac{1}{2} x_3^2 \mu(x_1) + F_{ext}(x_1,x_2) \right]\text{,} \\[4mm] 
      \displaystyle \dot{x}_3  & = & \displaystyle \frac{1}{L(x_1)} \left[u - Rx_3 + x_2x_3 \mu(x_1)\right]\text{.}
   \end{array}
   \right .
   \label{eq-systeme}
\end{equation}


\section{Nonlinear control for the uncertain model}

The proposed method is a combination of a backstepping control and a sliding mode control that makes the closed loop system states converge to a ball around the desired equilibrium point. This equilibrium point depends on the position reference signal $y_r$ the state $x_1$ has to track. The first subsection is dedicated to the design of a backstepping methodology to control the mechanical part of the system. The second subsection develops a sliding mode control to drive the coil current to the required value. Finally, a proof of the system convergence under the proposed control will then be detailed in the third subsection.

The controller design starts by the control of the position and velocity states. The choice of the backstepping method is natural due to the cascade form of the subsystem (\ref{eq-subsysteme}),
\begin{equation}
   \left \{
   \begin{array}{r c l c r c l}
      \displaystyle \dot{x}_1  & = & \displaystyle x_2  & \text{,} & \dot{x}_2  & = & \displaystyle \frac{1}{m}\left[ \frac{1}{2} x_{3d}^2 \mu + F_{ext}(x_1,x_2) \right]\text{.} 
   \end{array}
   \right .
   \label{eq-subsysteme}
\end{equation}
where $x_{3d}^2$ stands for the virtual control input.

\begin{theorem}
Consider $\alpha_1$, $\alpha_2$ two positives scalars, $a =  \displaystyle 1 - \alpha_1^2 + \frac{\lambda}{m} \alpha_1 - \frac{K}{m}$ and $b = \alpha_1 - \displaystyle \frac{\lambda}{m}  $. If the matrix $Q = \displaystyle \begin{pmatrix} -\alpha_1 & a/2\\ a/2 & b-\alpha_2 \end{pmatrix}$ is negative definite, then the virtual control law  $x_{3d}^2 = \displaystyle -\frac{2m}{\underline{\mu}(x_1+y_r)}\alpha_2 (x_2 + \alpha_1 x_1 - \alpha_1 y_r)$ makes the subsystem (\ref{eq-subsysteme}) convergent to a ball of center $  \begin{pmatrix} y_r \\ 0 \end{pmatrix}$ and of radius $ \frac{\delta}{\alpha \theta}$ , with $\delta = \frac{K}{m} \mid y_r \mid $ , $\alpha =\mid \lambda_{min} (Q) \mid$ and $\theta$ is a positive scalar lower than one.
\label{th_1}
\end{theorem}

\begin{proof}
Consider the classical change of variable

\begin{equation}
   \left \{
   \begin{array}{r c l}
      \displaystyle z_1  & = & x_1 - y_r \text{,}\\[4mm] 
      \displaystyle z_2  & = & x_2 + \alpha_1 z_1\text{,}
   \end{array}
   \right .
   \label{eq-var_error}
\end{equation}
with $\alpha_1$ a positive scalar. The subsystem (\ref{eq-subsysteme}) can be rewritten as: 
\begin{equation}
   \left \{
   \begin{array}{r c l}
      \displaystyle \dot{z}_1  & = &-\alpha_1 z_1 + z_2  \text{,} \\[4mm] 
      \displaystyle \dot{z}_2  & = & \displaystyle \frac{1}{m}\left[ \frac{1}{2} x_{3d}^2 \mu + F_{ext}(z_1,z_2) \right]+ \alpha_1z_2  - \alpha_1^2z_1 \text{.}
   \end{array}
   \right .
   \label{eq-subsysteme_z}
\end{equation}
In order to achieve the closed loop desired properties, let us consider a Lyapunov function of the form $V_1 = \frac{1}{2}z_1^2+ \frac{1}{2}z_2^2\text{.}$ The derivative of $V_1$ along the trajectories of (\ref{eq-subsysteme_z}) leads to: 

\begin{equation}
\dot{V}_1 = \displaystyle -\alpha_1 z_1^2 +a z_1 z_2  + b z_2^2  + z_2 \frac{1}{2m} x_{3d}^2 \mu -z_2\frac{K}{m}y_r\text{,}
\end{equation}
with $a =  \displaystyle 1 - \alpha_1^2 + \frac{\lambda}{m} \alpha_1 - \frac{K}{m}$ and $b = \alpha_1 - \displaystyle \frac{\lambda}{m}  $. Consider the control law  $ \displaystyle x_{3d}^2 = -\frac{2m}{\underline{\mu}}\alpha_2 z_2$ where $\alpha_2 > \mid b \mid$ is a positive scalar.  $\dot{V}_1$ can then be rewritten as: 
\begin{equation}
\begin{array}{r c l }
\dot{V}_1 & =  & \displaystyle -\alpha_1 z_1^2 +a z_1 z_2  + z_2^2 \left(b -\alpha_2 \frac{\mu}{\underline{\mu}}\right) -\frac{K}{m}z_2y_r\text{,}\\
\dot{V}_1 & \leq & \displaystyle -\alpha_1 z_1^2 +a z_1 z_2  + (b -\alpha_2 ) z_2^2 -\frac{K}{m}z_2y_r\text{,}\\
\dot{V}_1 & \leq & \displaystyle z^T Q z -\frac{K}{m}z_2y_r\text{,}
\end{array}
\end{equation}
with $ z= \begin{pmatrix} z_1 & z_2 \end{pmatrix} ^T $ and $\alpha_1$ and $\alpha_2$ are chosen to obtain $Q = \displaystyle \begin{pmatrix} -\alpha_1 & a/2\\ a/2 & b -\alpha_2  \end{pmatrix}$ a negative definite matrix. Therefore,
\begin{equation}
\dot{V}_1 \leq \displaystyle \lambda_{min} (Q) \mid \mid z \mid \mid ^2 +\frac{K}{m} \mid z_2 \mid y_r\text{.}
\end{equation}
Notice that the term $ -\frac{K}{m}y_r $ can be considered as a nonvanishing perturbation \cite{khalil1996noninear} and as $y_r$ is a bounded signal, there exists $\delta >0$ such that $ \mid \frac{K}{m}y_r  \mid \leq \delta$. Consider a scalar $\theta \in [0,1]$, it implies that 
\begin{equation}
\begin{array}{r c l}
\dot{V}_1  & \leq  & \displaystyle -\alpha \mid \mid z \mid \mid ^2 +\mid z_2 \mid \delta\\          
           & \leq  & -(1-\theta) \alpha \mid \mid z \mid \mid ^2 - \theta \alpha \mid \mid z \mid \mid^2  + \mid \mid z \mid \mid  \delta  \\
           & \leq  & -(1-\theta) \alpha \mid \mid z \mid \mid ^2 \qquad \forall \; \;\mid \mid z \mid \mid > \displaystyle \frac{\delta}{\alpha \theta}\text{.} \\
\end{array}
\end{equation}

Following \cite{khalil1996noninear}, the subsystem (\ref{eq-subsysteme}) converges to the disc of center $0$ and radius $ \frac{\delta}{\alpha \theta}$ which concludes the proof.

\end{proof}

\begin{Remark}
As the magnetic force $F_{mag}$ is uncertain, it cannot compensate exactly the constant term $ \displaystyle -\frac{K}{m}y_r$. The best we can do is to minimize its effect, leading to bound ultimately the subsystem by a small bound \cite{khalil1996noninear}.
\end{Remark}

\begin{Remark}
A more general Lyapunov function of the form $V_1 = z^TPz$ may be used in order to reduce the size of the disc. Notice also that an optimization scheme could be implemented in order to minimize the size of the ball in which the states $(x_1,x_2)$ converges.
\end{Remark}
The next step of the controller design is to design a control law such that $x_3$ converges to $x_{3d}$. Since the functions $L(x_1)$ and $\mu(x_1)$ are uncertain, we rely on a sliding mode approach of order 1 \cite{perruquetti2002sliding}.\newline

\begin{theorem}
Consider a scalar $\epsilon>0$, $\alpha_3 =   \displaystyle R \mid   x_{3d} \mid + \mid (z_2-\alpha_1z_1)(S + x_{3d}) \mid \bar{\mu} +  \mid\dot{x}_{3d} \mid \bar{L}(x_1) + \epsilon$, the control law $u = - \alpha_3 sign(S)$, with the sliding surface $S = x_3 - x_{3d}$ makes $x_3$ converge in finite time towards $x_{3d}$.
\end{theorem}

\begin{proof}
Consider the sliding surface $S = x_3 - x_{3d}$ and the Lyapunov function $V_2 = \frac{1}{2}S^2$. The derivative of $V_2$ along the trajectories of (\ref{eq-systeme}) leads to:

\begin{equation}
\begin{array}{r c l}
\dot{V}_2   & =  & \displaystyle S \frac{1}{L(x_1)} u    + S \Big( \frac{1}{L(x_1)} \left[- R (S + x_{3d}) \right.  \\[3mm]
            &    & \displaystyle \Big. \left. + (z_2-\alpha_1z_1)(S + x_{3d})  \mu  \right] - \dot{x}_{3d} \Big)\text{.}
\end{array}
\end{equation}
Let us choose $ u =  - \alpha_3 sign(S)$ with a gain $\alpha_3 > 0$ then 

\begin{equation}
\begin{array}{r c l }
\dot{V}_2   & = & \displaystyle -\alpha_3 \mid S \mid \frac{1}{L(x_1)}    - R S^2\frac{1}{L(x_1)} -  S  \dot{x}_{3d}  \\
 &  +  & \displaystyle  S \Big( \frac{1}{L(x_1)} \left[- R  x_{3d} +(z_2-\alpha_1z_1)(S + x_{3d})  \mu  \right]  \Big)\text{.}
\end{array}
\end{equation}
Notice that $ \displaystyle \frac{1}{L(x_1)}$ is an uncertain but strictly positive function and therefore,
\begin{equation}
\begin{array}{r c l}
\dot{V}_2 	  &  \leq & \displaystyle  \mid S \mid   \Big| \left( \frac{1}{L(x_1)} \left[-\alpha_3 - R x_{3d} \right. \right. \\
              &       & \Big. \left. + (z_2-\alpha_1z_1)(S + x_{3d})  \mu -  \dot{x}_{3d}L(x_1) \right] \Big) \Big|\text{.}
\end{array}
\end{equation}
Setting 
\begin{equation}
\alpha_3 =   \displaystyle R \mid   x_{3d} \mid + \mid (z_2-\alpha_1z_1)(S + x_{3d}) \mid \bar{\mu} +  \mid\dot{x}_{3d} \mid \bar{L}(x_1) + \epsilon\text{,}
\label{eq-a3}
\end{equation}
where $\epsilon > 0$, we obtain $\dot{V}_2 \leq - \epsilon S = - \epsilon \sqrt{V_2}$ which proves the convergence in finite time of $x_3$ towards $x_{3d}$.
\end{proof}
The last step consists in proving the convergence of the whole system (\ref{eq-systeme}) to a ball around the desired equilibrium point with the control laws defined in Theorem 1 and Theorem 2

\begin{theorem}
Consider $\epsilon_1$, $\alpha_1$, $\alpha_2$ three positives scalars, $a =  \displaystyle 1 - \alpha_1^2 + \frac{\lambda}{m} \alpha_1 - \frac{K}{m}$ and $b = \alpha_1 - \displaystyle \frac{\lambda}{m}  $, $Q = \displaystyle \begin{pmatrix} -\alpha_1 & a/2\\ a/2 & b-\alpha_2 \end{pmatrix}$,$\alpha = \mid\lambda_{min} (Q)\mid$,$\alpha_3 =   \displaystyle R \mid   x_{3d} \mid + \mid (z_2-\alpha_1z_1)(S + x_{3d}) \mid \bar{\mu} +  \mid\dot{x}_{3d} \mid \bar{L}(x_1) + \epsilon$, and $\theta$ a positive scalar lower than one. If the matrix $Q$ is negative definite, then the control law  $u = -\alpha_3 sign (S)$ makes the system (\ref{eq-systeme}) convergent to a disc of center $  \begin{pmatrix} y_r \\ 0 \\ x_{3d}\end{pmatrix}$ of radius $ \frac{\delta}{\alpha \theta}$, belonging the map $x_3 = x_{3d}$.
\end{theorem}

\begin{proof}
Using $V = V_1 + V_2$ , we proved in Theorem 1 and Theorem 2 that 
\begin{equation}
\begin{array}{r c l}
\dot{V}   &     \leq   & -\alpha \mid \mid z \mid \mid ^2 + \delta \mid \mid z \mid \mid  - \epsilon \mid S \mid  + \frac{2}{m}\mid S \mid (x_3+x_{3d})\bar{\mu} \text{.}\\
\end{array}
\end{equation}
Taking $ \displaystyle \epsilon = \frac{2}{m}(x_3+x_{3d})\bar{\mu} + \epsilon_1$ with $\epsilon_1 >0$
\begin{equation}
\begin{array}{r c l}
\dot{V}    &   \leq     & -\alpha \mid \mid z \mid \mid ^2 + \delta \mid \mid z \mid \mid  - \epsilon_1 \mid S \mid  \text{,}\\
           & \leq  & -(1-\theta) \alpha \mid \mid z \mid \mid ^2  - \epsilon_1 \mid S \mid\; \text{,}\;\; \forall \mid \mid z \mid \mid > \displaystyle \frac{\delta}{\alpha \theta} \text{.}\\
\end{array} 
\end{equation}
As $S$ converges to $0$ in finite time we prove that the whole system converges asymptotically to an invariant and attractive disc of center $  \begin{pmatrix} y_r \\ 0 \\ x_{3d}\end{pmatrix}$ of radius $ \frac{\delta}{\alpha \theta}$, belonging the map $x_3 = x_{3d}$.
\end{proof}
\begin{Remark}
The proposed approach is a combined backstepping - sliding mode control. Notice that a backstepping approach cannot be applied all along the design due to the nonvanishing term depending on a spring force which cannot be compensated by the control because of the uncertainty.
In the same way, it appears to be complicated to design a sliding mode from the beginning because the nonvanishing term does not satisfy the matching condition, the usual condition allowing a sliding mode control to reject a perturbation \cite{perruquetti2002sliding}.
\end{Remark}

\section{Simulation}
Let us consider the EMA whose parameters are given in Table \ref{table-param-EMA} and consider the control defined by Theorem 3 with the followings parameters   $\alpha_1 = 10$, $\alpha_2 = 20000$ and $\epsilon_1 = 10$.
\begin{table}[h]
\begin{center}
\begin{tabular}{|c|c|c|}
  \hline
  Name  & Value  & Description   \\
  \hline
  $\rho_x$ & 2.8 $\times 10^10 $ $H^{-1}m^{-1}$⋅ &  Airgap reluctance\\
  \hline
  $\rho_0$ & 630 $H^{-1}$ & Magnetic circuit reluctance\\ 
  \hline 
  $\lambda$ & 5 $Nm^{-2}$ & Friction coefficient \\
  \hline
  $K$   & 120 $Nm^{-1}$ & Spring constant\\
  \hline
  $N$ & 70  & Coil winding\\
  \hline
  $m$ & 0.1 $kg$ & Mass of the moving part\\
  \hline
  $R$ & 0.4 $\Omega	$  & coil intern resitor\\
  \hline
  $x_0$ & (0.001 $m$,0,0)$^T$ & Initial state\\
  \hline
\end{tabular}
\caption {}
\label{table-param-EMA}
\end{center}
\end{table} 
\vspace{-0.5cm}
In this simulation, $y_r$ is a step signal of 3 $mm$ amplitude.
The results of the position tracking simulation are shown in the Figure \ref{fig-resultats_x2}, the speed, the current and the input voltage $u$ respectively in Figures \ref{fig-resultats_x2}, \ref{fig-resultats_x3} and \ref{fig-resultats_u}.

\begin{figure}[h!]
\begin{center}
\includegraphics[trim={0mm 67mm 0mm 88mm}, clip,,width=8.8cm]{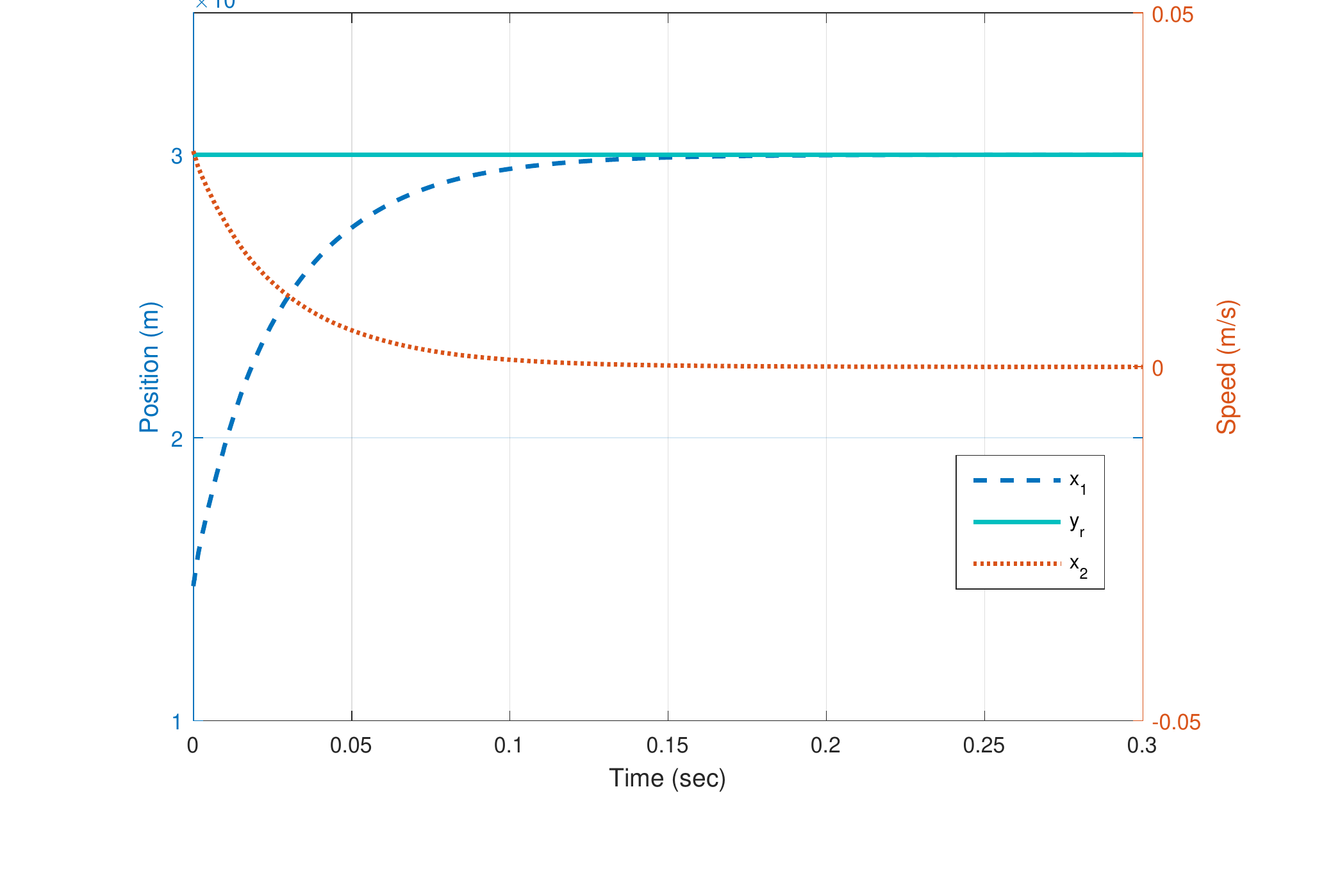} 
\end{center}
\caption{Evolution of the position tracking $x_1$ and velocity $x_2$}
\label{fig-resultats_x2}
\end{figure}
\begin{figure}[h!]
\begin{center}
\includegraphics[trim={2mm 83mm 8mm 75mm}, clip,,width=6cm]{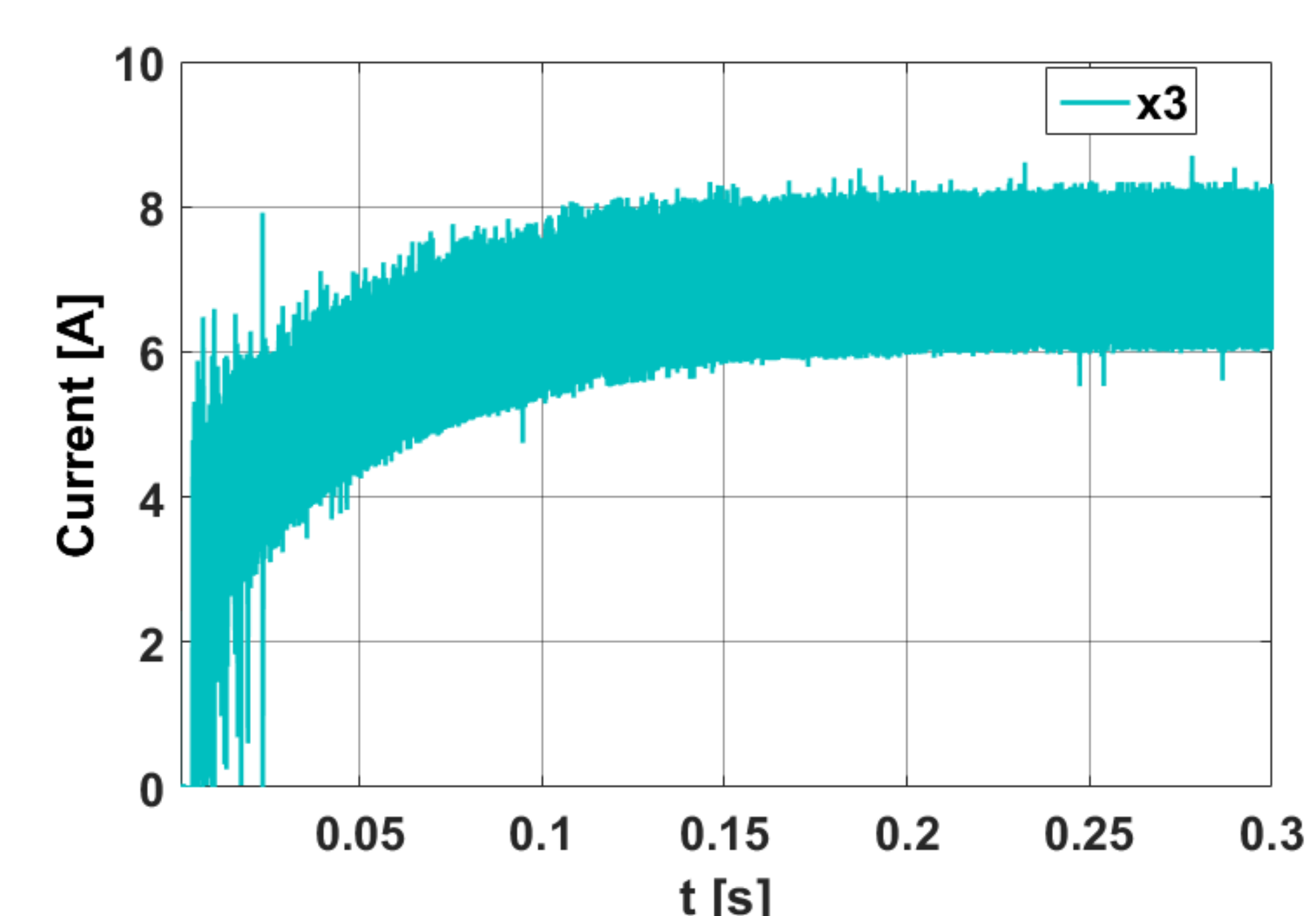} 
\end{center}
\caption{Evolution of the coil current $x_3$}
\label{fig-resultats_x3}
\end{figure}
\begin{figure}[h!]
\begin{center}
\includegraphics[trim={2mm 83mm 8mm 75mm}, clip,,width=6cm]{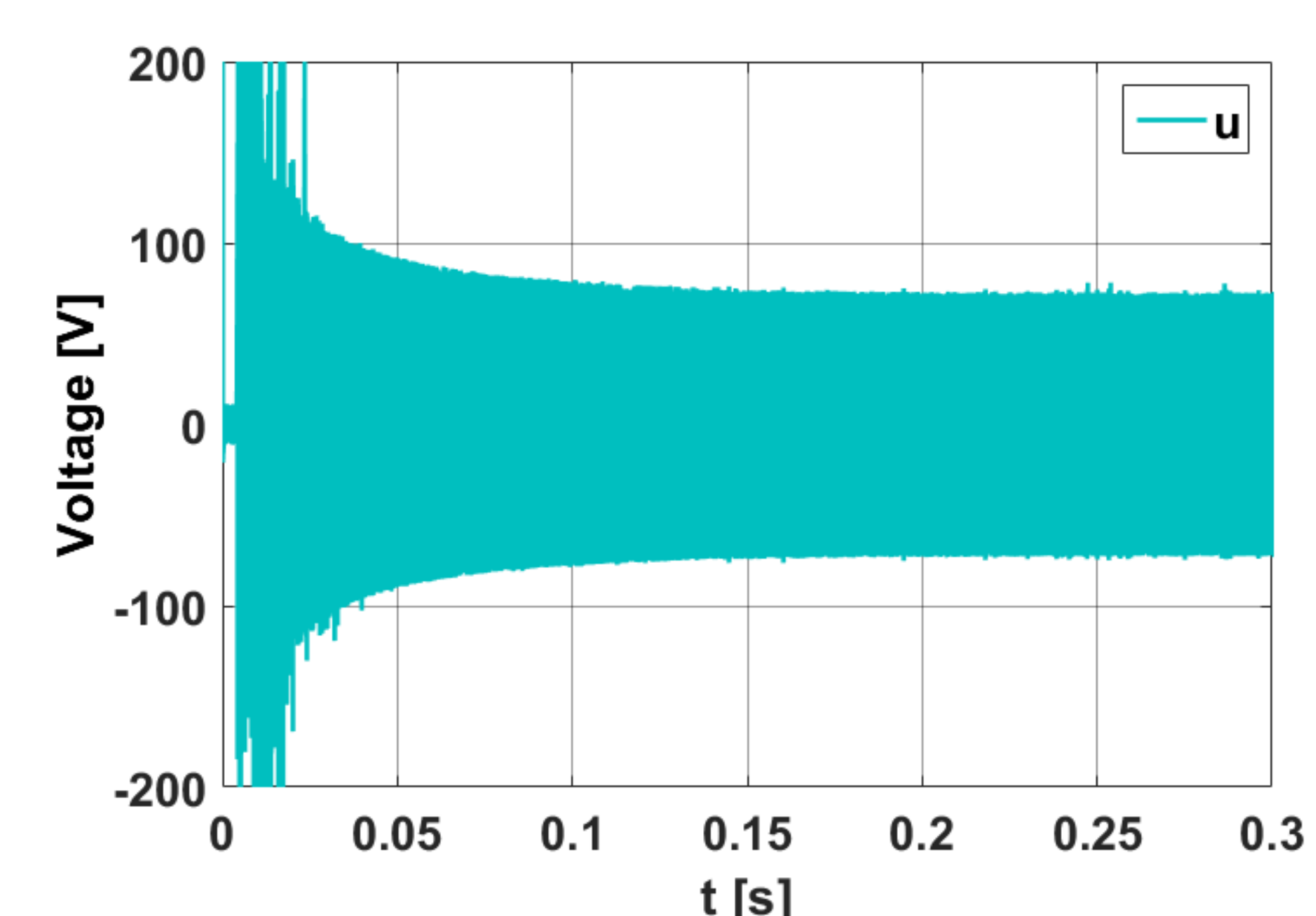} 
\end{center}
\caption{Evolution of the voltage input $u$}
\label{fig-resultats_u}
\end{figure}

The system position follows the reference $y_r$ with a good settling time ($ \approx 0.2s$) and no overshoot. There is no chattering on the position and the speed due to the pure integrator in the system. However, a chattering phenomenon appears on the current due to the discontinuous control $u$. Notice that from a practical point of view, this high frequency switching is not an issue since the implemented control is based on switching transistors. The coil current magnitude is standard as common EMA needs values of a few milliampere to a few Ampere. Notice that the control signal $u$ has also standard values. In this case it can be easily generated by a $150V$ switching power supply.


The Figure \ref{fig-resultats_z1z2} shows the error due to the uncertainty on the model for different initial values. As mentioned in Theorem 3, it is impossible to compensate the constant term $ \displaystyle -\frac{K}{m}y_r$ so the system converges around the reference but a small error remains. However, we may consider this algorithm as robust since the error is minimized. $S$ converges to $0$ and there is less than $0.1$\% error between $x_1$ and $y_r$.
\vspace{-0.5cm}
\begin{figure}[!h]
\begin{center}
\includegraphics[trim={12mm 95mm 15mm 80mm}, clip,width=0.48\textwidth]{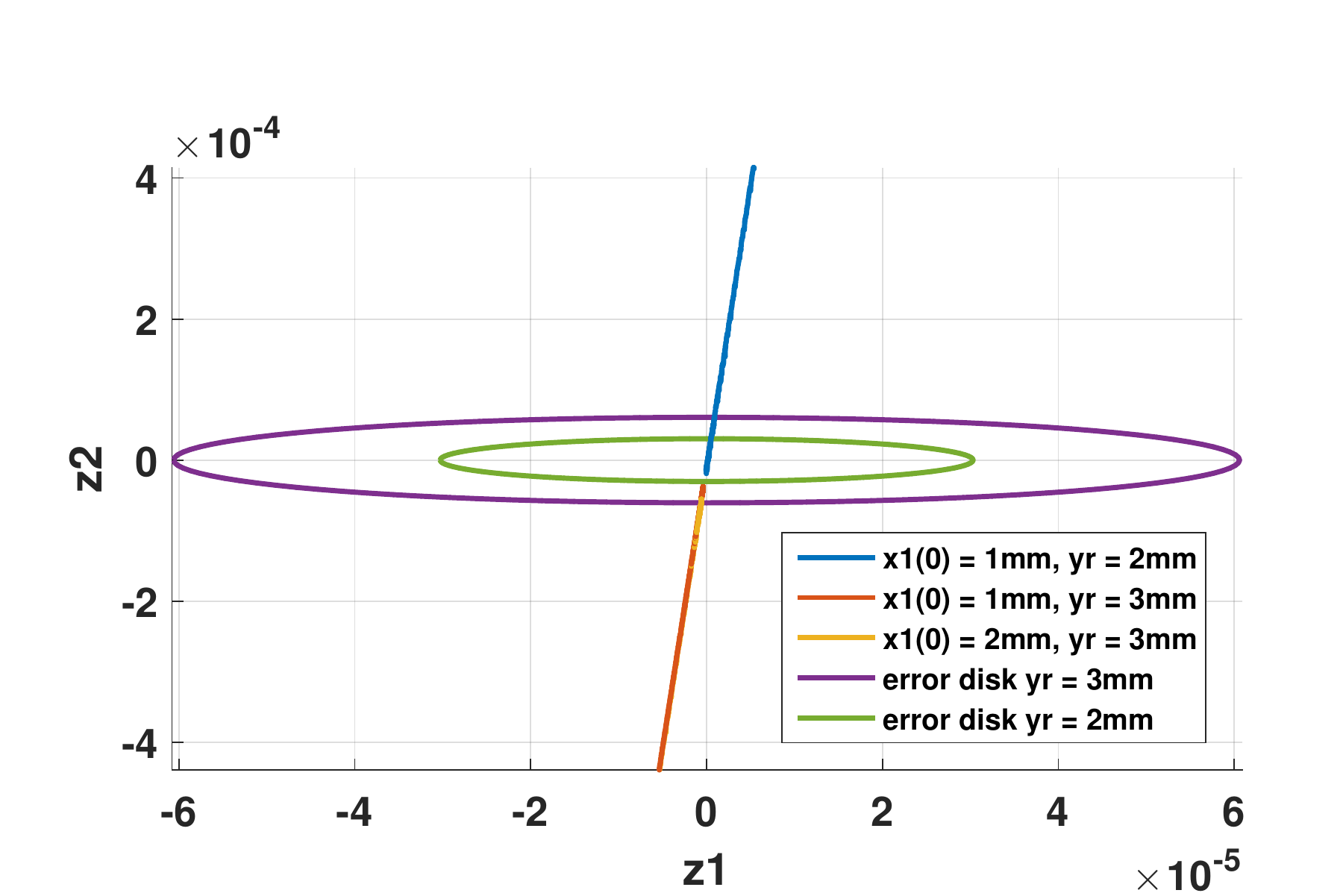} 
\end{center}
\caption{Zoom on the tracking error: Plane $z_1$ - $z_2$}
\label{fig-resultats_z1z2}
\end{figure}

\section{CONCLUSION}
This paper focuses on a robust control law design for an uncertain model of EMA which takes into account the flux fringing. The proposed control law relies on a combined backstepping and sliding mode control and ensures that the states are ultimately bounded within a disc centered around the reference to be tracked.

Future work consists in, on one hand validating the control law on a testbed, on the other hand, including LMI optimization algorithm in order to minimize the set in which the states converge. In addition, a future work should include the design of a control law taking also into account the magnetic saturation. 


\section*{ACKNOWLEDGMENT}
The Authors would like to thank Bruno Vieille from CNES and Fran\c{c}ois Dugu\'e from CSTM for the grants (CNES-CSTM contract number 170890) that partly supports this activity.

%
%

\bibliographystyle{unsrt}
\bibliography{bib}

\end{document}